\documentclass[journal,onecolumn]{IEEEtran}
\usepackage[normalem]{ulem}
\usepackage{url}
\usepackage{cite}

\ifCLASSINFOpdf

  \usepackage[pdftex]{graphicx}
  \DeclareGraphicsExtensions{.pdf,.jpg,.png}
  \graphicspath{C:/Users/Amichai/Desktop/BICA transactions/}
\else
\fi
\usepackage{caption}
\usepackage{amssymb}
\usepackage[cmex10]{amsmath}

\usepackage{mathtools}
\usepackage{array}
\usepackage{algorithm}
\usepackage{algorithmic}
\usepackage{xfrac }

\newtheorem{theorem}{Theorem}

\newenvironment{proof}[1][Proof]{\begin{trivlist}
\item[\hskip \labelsep {\bfseries #1}]}{\end{trivlist}}

\newcommand{\qed}{\nobreak \ifvmode \relax \else
      \ifdim\lastskip<1.5em \hskip-\lastskip
      \hskip1.5em plus0em minus0.5em \fi \nobreak
      \vrule height0.75em width0.5em depth0.25em\fi}

\newcommand{\RNum}[1]{\uppercase\expandafter{\romannumeral #1\relax}}

\begin{document}
%
\title{Generalized Independent Component Analysis Over Finite Alphabets}
%
%
%

\author{Amichai~Painsky,~\IEEEmembership{Member,~IEEE,}
        Saharon~Rosset and~Meir~Feder,~\IEEEmembership{Fellow,~IEEE}

\thanks{A. Painsky and S. Rosset are with the Statistics Department, Tel Aviv University, Tel Aviv, Israel. contact: amichaip@eng.tau.ac.il}
\thanks{M. Feder is with the Department of Electrical Engineering, Tel Aviv University, Tel Aviv, Israel}
\thanks{The material in this paper was presented in part at the IEEE International Symposium on Information Theory (ISIT) 2014.}}

%
%

\markboth{IEEE TRANSACTIONS ON INFORMATION THEORY}%
{Shell \MakeLowercase{\textit{et al.}}: Bare Demo of IEEEtran.cls for Journals}
%



\maketitle

\begin{abstract}
Independent component analysis (ICA) is a statistical method for transforming an observable multidimensional random
vector into components that are as statistically independent as possible from each other.
Usually the ICA framework assumes a model according to which the observations are generated (such as a linear transformation with additive noise). ICA over finite fields is a special case of ICA in which both the observations and the independent components are over a finite alphabet. In this work we consider a generalization of this framework in which an observation vector is decomposed to its independent components (as much as possible) with no prior assumption on the way it was generated. This generalization is also known as Barlow's \textit{minimal redundancy representation} problem and is considered an open problem. We propose several theorems and show that this NP hard problem can be accurately solved with a branch and bound search tree algorithm, or tightly approximated with a series of linear problems. Our contribution provides the first efficient and constructive set of solutions to Barlow's problem.
The minimal redundancy representation (also known as \textit{factorial code}) has many applications, mainly in the fields of Neural Networks and Deep Learning. The Binary ICA (BICA) is also shown to have applications in several domains including medical diagnosis, multi-cluster assignment, network tomography and internet resource management. In this work we show this formulation further applies to multiple disciplines in source coding such as predictive coding, distributed source coding and coding of large alphabet sources.  

\end{abstract}


%
\IEEEpeerreviewmaketitle

\section{Introduction}
%
%
%
%
\IEEEPARstart{I}{ndependent}  Component Analysis (ICA) addresses the recovery of unobserved statistically independent source signals from their observed mixtures, without full prior knowledge of the mixing function or the statistics of the source signals. The classical Independent Components Analysis framework usually assumes linear combinations of the independent sources over the field of real valued numbers $\mathbb{R}$ \cite{hyvarinen2004independent}. A special variant of the ICA problem is when the sources, the mixing model and the observed signals are over a finite field.

Several types of generative mixing models can be assumed when working over GF(P), such as modulu additive operations, OR operations (over the binary field) and others. Existing solutions to ICA mainly differ in their assumptions of the generative mixing model, the prior distribution of the mixing matrix (if such exists) and the noise model. The common assumption to these solutions is that there exist statistically independent source signals which are mixed according to some known generative model. 

In this work we drop this assumption and consider a generalized approach which is applied to a random vector and decomposes it into independent components (as much as possible) with no prior assumption on the way it was generated. This problem was first introduced by Barlow \cite{barlow1989finding} and is considered a long--standing open problem.


The rest of this manuscript is organized as follows: In Section \RNum{2} we review the work that was previously done in finite fields ICA and factorial codes. In Section \RNum{3} we present the generalized binary ICA problem, propose several theorems and 
introduce three different algorithms for it. In Section \RNum{4} we extend our discussion to any finite alphabet, focusing on the increased computational load it confers and suggest methods of dealing with it. We then dedicate Section \RNum{5} to a constrained version of our problem, aimed to further reduce the computational complexity of our suggested solutions. In Section \RNum{6} we present several applications for our suggested framework, mainly focusing on the field of source coding.

\section{Previous Work}

In his work from 1989, Barlow \cite{barlow1989finding} presented a \textit{minimally redundant  representation} scheme for binary data. He claimed that a good representation should capture and remove the redundancy of the data. This leads to a \textit{factorial representation/ encoding} in which the components are as mutually independent of each other as possible. Barlow suggested that such representation may be achieved through \textit{minimum entropy encoding}: an invertible transformation (i.e., with no information loss) which minimizes the sum of marginal entropies (as later presented in (\ref{eq:sum_ent_min})).

Factorial representations have several advantages. The probability of the occurrence of any realization can be simply computed as the product of the probabilities of the individual components that represent it (assuming such decomposition exists). In addition, any method of finding factorial codes automatically implements \textit{Occam's razor} which prefers simpler models over more complex ones, where simplicity is defined as the number of parameters necessary to represent the joint distribution of the data.
In the context of supervised learning, independent features can also make later learning easier; if the input units to a supervised learning networks are uncorrelated, than the Hessian of its error function is diagonal, allowing accelerated learning abilities \cite{becker1988improving}.
There exists a large body of work which demonstrates the use of factorial codes in learning problems. This mainly includes Neural Networks \cite{becker1996unsupervised, obradovic1996information} with application to facial recognition \cite{choi2000factorial, bartlett2002face, martiriggiano2005face, bartlett2007information}  and more recently, Deep Learning \cite{schmidhuber2011fast, schmidhuber2015deep}.  

Unfortunately Barlow did not propose any direct method for finding factorial codes. Atick and Redlich \cite{atick1990towards} proposed a cost function for Barlow's principle for linear systems, which minimize the redundancy of the data subject to a minimal information loss constraint. This is closely related to Plumbey's \cite{plumbley1993efficient}  objective function, which minimizes the information loss subject to a fixed redundancy constraint. Schmidhuber \cite{schmidhuber1992learning} then proposed several ways of approximating Barlow's minimum redundancy principle in the non--linear case. This naturally implies much stronger results of statistical independence. However, Schmidhuber's scheme is rather complex, and appears to be subject to local minima \cite{becker1996unsupervised}. To our best knowledge, the problem of finding minimum entropy codes is still considered an open problem. In this work we present what appears to be the first efficient and constructive set of methods for minimizing Barlow's redundancy criterion. 

In a second line of work, we may consider our contribution as a generalization of the BICA problem. In his pioneering BICA work, Yeredor \cite{yeredor2007ica} assumed linear XOR mixtures and investigated the identifiability problem. A deflation algorithm is proposed for source separation based on entropy minimization. Yeredor assumes the number of independent sources is known and the mixing matrix is a $K$-by-$K$ invertible matrix. Under these constraints, he proves that the XOR model is invertible and there exists a unique transformation matrix to recover the independent components up to permutation ambiguity. 

Yeredor then extended his work \cite{yeredor2011independent} to cover the ICA problem over Galois fields of any prime number. His ideas were further analyzed and improved by Gutch et al. \cite{gutch2012ica}. 


In \cite{vsingliar2006noisy}, a noise-OR model is introduced to model dependency among observable random variables using $K$ (known) latent factors. A variational inference algorithm is developed. In the noise-OR model, the probabilistic dependency between observable vectors and latent vectors is modeled via the noise-OR conditional distribution. 

In \cite{wood2012non}, the observations are generated from a noise-OR generative model. The prior of the mixing matrix is modeled as the Indian buffet process \cite{griffiths2005infinite}. Reversible jump Markov chain Monte Carlo and Gibbs sampler techniques are applied to determine the mixing matrix

Streith et al. \cite{streich2009multi} study the BICA problem  where the observations are either drawn from a signal following OR mixtures or from a noise component. The key assumption made in that work is that the observations are conditionally independent given the model parameters (as opposed to the latent variables). This greatly reduces the computational complexity and makes the scheme amenable to a objective descent-based optimization solution. However, this assumption is in general invalid.

In \cite{nguyen2011binary}, Nguyen and Zheng consider OR mixtures and propose a deterministic iterative algorithm
to determine the distribution of the latent random variables and the mixing matrix. 
    
There also exists a large body of work on blind deconvolution with binary sources in the context of wireless communication \cite{diamantaras2006blind,yuanqing2003blind} and some literature on Boolean/binary factor analysis (BFA) which is also related to this topic \cite{belohlavek2010discovery}.

\section{Generalized Binary Independent Component Analysis}

Throughout this paper we use the following standard notation: underlines denote vector quantities, where their respective components are written without underlines but with index. For example, the components of the $n$-dimensional vector $\underline{X}$ are $X_1, X_2, \dots X_n$.
Random variables are denoted with capital letters while their realizations are denoted with the respective lower-case letters. $P_{\underline{X}}\left(\b{x}\right)  \triangleq P(X_1 = x_1, X_2= x_2\dots) $ is the probability function of  $\underline{X}$ while $H\left(\underline{X}\right)$ is the entropy of $\underline{X}$. This means  $H\left(\underline{X}\right)=-\sum_{\b{x}} P_{\underline{X}}\left(\b{x}\right) \log{P_{\underline{X}}\left(\b{x}\right)}$ where the $\log{}$ function denotes a logarithm of base $2$ and $\lim_{x \to 0} x\log{(x)} = 0$. Further, we denote the binary entropy as $h_b (p)=-p\log{p}-(1-p)\log{(1-p)}$.

\subsection{Problem Formulation}
Suppose we are given a random vector $\underline{X}\sim P_{\b{x}}\left(\b{x}\right)$ of dimension $n$ and alphabet size $A$ for each of its components. We are interested in an invertible, not necessarily linear, transformation $\underline{Y}=g(\underline{X})$ such that $\underline{Y}$  is of the same dimension and alphabet size, $g:\{1,\dots,A\}^n \rightarrow \{1,\dots,A\}^n$. In addition we would like the components of $\underline{Y}$ to be as "statistically independent as possible".

The common ICA setup is not limited to invertible transformations (hence $\underline{Y}$ and  $\underline{X}$ may be of different dimensions). However, in our work we focus on this setup as we would like $\underline{Y}=g(\underline{X})$ to be ``lossless" in the sense that we do not loose any information. Further motivation to this setup is discussed in \cite{barlow1989finding, schmidhuber1992learning} and throughout the Applications section below.    

Notice that an invertible transformation of a vector $\underline{X}$, where $\underline{X}$ is over a finite alphabet of size $A$, is actually a one-to-one mapping (i.e., permutation) of its $A^n$ words. For example, if $\underline{X}$ is over a binary alphabet and is of size $n$, then there are $2^n!$ possible permutations of its words. 

To quantify the statistical independence among the components of the vector $\underline{Y}$ we use the well-known total correlation measure, which was first introduced by Watanabe \cite{watanabe1960information} as a multivariate generalization of the mutual information.

\begin{equation}\label{eq:min_criterion}
C(\underline{Y})={\displaystyle \sum_{i=1}^{n}{H(Y_i)}-H(\underline{Y})}.
\end{equation}
This measure can also be viewed as the cost of coding the vector $\underline{Y}$ component-wise, as if its components were statistically independent, compared to its true entropy. Notice that the total correlation is non-negative and equals zero iff the components of $\underline{Y}$ are mutually independent. Therefore, ``as statistically independent as possible" may be quantified by minimizing $C(\underline{Y})$. The total correlation measure was considered as an objective for minimal redundency representation by Barlow \cite{barlow1989finding}. It is also not new to finite field ICA problems, as demonstrated in \cite{attux2011immune}. Moreover, we show that it is specifically adequate to our applications, as described in the last section.

Since we define $\underline{Y}$ to be an invertible transformation of $\underline{X}$ we have $H(\underline{Y})=H(\underline{X})$ and our minimization objective is

\begin{equation}
{\displaystyle \sum_{i=1}^{n}{H(Y_i)} \rightarrow min.}
\label{eq:sum_ent_min}
\end{equation}

In this section we focus on the binary case. The probability function of the vector  $\underline{X}$  is therefore defined by $P(X_1,\ldots,X_n)$ over $2^n=N$ possible words and our objective function is simply

\begin{equation}
{\displaystyle \sum_{i=1}^{n}{h_b(P(Y_i=0))} \rightarrow min.}
\label{eq:sum_ent_min_binary}
\end{equation}

We notice that $P(Y_i=0)$ is the sum of probabilities of all words whose $i^{th}$ bit equals $0$. We further notice that the optimal transformation is not unique. For example, we can always invert the  $i^{th}$ bit of all words, or shuffle the bits, to achieve the same minimum.

Any approach which exploits the full statistical description of the joint probability distribution
of $\underline{X}$  would require going over all $2^n$ possible words at least once. Therefore, a computational load of at least $O(2^n)$ seems inevitable. Still, this is significantly smaller (and often realistically far more affordable) than $O(2^n!)$, required by brute-force search over all possible permutations. 

Indeed, the complexity of the currently known binary ICA (and Factorial codes) algorithms falls within this range. The AMERICA \cite{yeredor2011independent} algorithm, which assumes a XOR mixture, has a complexity of $O(n^2\cdot 2^n)$. The MEXICO algorithm, which is an enhanced version of AMERICA, achieves a complexity of $O(2^n)$ under some restrictive assumptions on the mixing matrix. In \cite{nguyen2011binary} the assumption is that the data was generated over OR mixtures and the asymptotic complexity is $O(n \cdot 2^n)$. There also exist other heuristic methods which avoid an exhaustive search, such as \cite{attux2011immune} for BICA or \cite{schmidhuber1992learning} for Factorial codes. These methods, however, do not guarantee convergence to the global optimal solution.

Looking at the BICA framework, we notice two fundamental a-priori assumptions:
\begin{enumerate}

\item	The vector $\underline{X}$ is a mixture of independent components and there exists an inverse transformation which decomposes these components.
\item	The generative model (linear, XOR field, etc.) of the mixture function is known.

\end{enumerate}

In this work we drop these assumptions and solve the ICA problem over finite alphabets with no prior assumption on the vector $\underline{X}$. 
As a first step towards this goal, let us drop Assumption $2$ and keep Assumption $1$, stating that underlying independent components do exist. The following combinatorial algorithm proves to solve this problem, over the binary alphabet, in $O(n \cdot 2^n)$ computations.

\subsection{Generalized BICA with Underlying Independent Components}

In this section we assume that underlying independent components exist. In other words, we assume there exists a permutation $\underline{Y}=g(\underline{X})$  such that the vector $\underline{Y}$ is statistically independent $P(Y_1,\ldots,Y_n)=\prod_{i=1}^{n}{P(Y_i)}$. Denote the marginal probability of the $i^{th}$ bit equals $0$ as $\pi_i=P(Y_i=0)$. Notice that by possibly inverting bits we may assume every $\pi_i$ is at most $\frac{1}{2}$ and by reordering we may have,
without loss of generality, that $\pi_n \leq  \pi_{n-1} \leq  \cdots  \leq  \pi_1 \leq 1/2$. In addition, we assume a non-degenerate setup where $\pi_n>0$. For simplicity of presentation, we first analyze the case where $\pi_n <  \pi_{n-1} <  \cdots  <  \pi_1 \leq 1/2$. This is easily generalized to the case where several $\pi_i$ may equal, as discussed later in this section. 

Denote the $2^n$ probabilities of $P(\underline{Y}=\underline{y} )$
as $p_1, p_2, \ldots , p_N$, assumed to be ordered so that $p_1 \leq p_2 \leq \cdots \leq 
p_N$.
We first notice that the probability of the all-zeros word, $P(Y_n=0, Y_{n-1}=0, \ldots ,Y_1=0)=\prod_{i=1}^{n}{\pi_i}$  is the smallest possible probability since all parameters are not greater than 0.5. Therefore we have $p_1=\prod_{i=1}^{n}{\pi_i}$.

Since $\pi_1$ is the largest parameter of all $\pi_i$, the second smallest probability is just $P(Y_n=0,\ldots,Y_2=0,Y_1=1)=\pi_n \cdot \pi_{n-1} \cdot \ldots \cdot \pi_2 \cdot (1-\pi_1 )=p_2$. Therefore we can recover $\pi_1$ from $\frac{1-\pi_1}{\pi_1} =\frac{p_2}{p_1}$,  leading to $\pi_1=\frac{p_1}{p_1+p_2}$.

We can further identify the third smallest probability as $p_3=\pi_n \cdot \pi_{n-1} \cdot \ldots \cdot \pi_3 \cdot (1-\pi_2) \cdot \pi_1$. This leads to $\pi_2=\frac{p_1}{p_1+p_3}$.

However, as we get to $p_4$ we notice we can no longer definitely identify its components; it may either equal $\pi_n \cdot \pi_{n-1} \cdot \ldots \cdot \pi_3 \cdot (1-\pi_2) \cdot (1-\pi_1)$ or $\pi_n \cdot \pi_{n-1} \cdot \ldots \cdot (1-\pi_3 )\cdot \pi_2 \cdot \pi_1$. This ambiguity is easily resolved since we can compute the value of $\pi_n \cdot \pi_{n-1} \cdot \ldots \cdot \pi_3 \cdot (1-\pi_2) \cdot (1-\pi_1)$ from the parameters we already found and compare it with $p_4$. Specifically, If $\pi_n \cdot \pi_{n-1} \cdot \ldots \cdot \pi_3 \cdot (1-\pi_2) \cdot (1-\pi_1) \neq p_4$ then we necessarily have $\pi_n \cdot \pi_{n-1} \cdot \ldots \cdot (1-\pi_3 )\cdot \pi_2 \cdot \pi_1 = p_4$ from which we can recover $\pi_3$. Otherwise $\pi_n \cdot \pi_{n-1} \cdot \ldots \cdot (1-\pi_3 )\cdot \pi_2 \cdot \pi_1 = p_5$ from which we can again recover $\pi_3$ and proceed to the next parameter.

Let us generalize this approach. Denote $\Lambda_k$ as a set of probabilities of all words whose $(k+1)^{th}, \dots, n^{th}$ bits are all zero.
\begin{theorem} 
\label{theorem1}
Let $j$ be an arbitrary index in $\{1,2,\dots,N\}$.
Assume we are given that $p_j$, the $j^{th}$ smallest probability in the set of probabilities $P(Y_n,\ldots,Y_1)$, satisfies the following decomposition 
\begin{equation}\nonumber
p_j=\pi_n \cdot \pi_{n-1} \cdot \ldots \cdot \pi_{k+1} \cdot (1-\pi_k) \cdot \pi_{k-1}  \cdot \ldots \cdot ∙\pi_1.
\end{equation}
 Further assume the values of $\Lambda_{k-1}$ are all given in a sorted manner. Then the complexity of finding the value of $\pi_n \cdot \pi_{n-1} \cdot \ldots \cdot \pi_{k+2} \cdot (1-\pi_{k+1}) \cdot \pi_k \cdot \ldots \cdot ∙\pi_1$, and calculating and sorting the values of $\Lambda_k$  is $O(2^k)$.
\end{theorem}

\begin{proof} 
Since the values of  $\Lambda_{k-1}$ and $\pi_k$ are given we can calculate the values which are still missing to know $\Lambda_k$ entirely by simply multiplying each element of $\Lambda_{k-1}$ by $\frac{1-\pi_k}{\pi_k}$.  Denote this set of values as $\bar{\Lambda}_{k-1}$. Since we assume the set $\Lambda_{k-1}$ is sorted then $\bar{\Lambda}_{k-1}$ is also sorted and the size of each set is $2^{k-1}$. Therefore, the complexity of sorting $\Lambda_{k}$ is the complexity of merging two sorted lists, which is $O(2^k)$.

In order to find the value of $\pi_n \cdot \pi_{n-1} \cdot \ldots \cdot \pi_{k+2} \cdot (1-\pi_{k+1}) \cdot \pi_k \cdot \ldots \cdot ∙\pi_1$ we need to go over all the values which are larger than $p_j$ and are not in $\Lambda_k$. However, since both the list of all probabilities $\{p_i\}_{i=1}^{N}$ and the set $\Lambda_k$ are sorted we can perform a binary search to find the smallest entry for which the lists differ. The complexity of such search is  $O(\log⁡{(2^k)})=O(k)$ which is smaller than $O(2^k)$. Therefore, the overall complexity is $O(2^k)$  $\blacksquare$
\end{proof}

Our algorithm is based on this theorem. We initialize the values $p_1, p_2$ and $\Lambda_{1}$, and   
for each step $k=3 \ldots n$ we calculate $\pi_n \cdot \pi_{n-1} \cdot \ldots \cdot \pi_{k+1} \cdot (1-\pi_{k}) \cdot \pi_{k-1} \cdot \ldots \cdot ∙\pi_1$ and $\Lambda_{k-1}$. 

The complexity of our suggested algorithm is therefore $\sum_{k=1}^{n}{O(2^k)}=O(n\cdot2^n)$. 
However, we notice that by means of the well-known quicksort algorithm \cite{Hoare:1961:AQ:366622.366644}, the complexity of our preprocessing sorting phase is $O(N\log{(N)})=O(n\cdot2^n)$.
Therefore, in order to find the optimal permutation we need $O(n\cdot2^n)$ for sorting the given probability list and $O(2^n)$ for extracting the parameters of $P(X_1, \ldots ,X_n)$.
Notice that this algorithm is combinatorial in its essence and is not robust when dealing with real data. In other words, the performance of this algorithm strongly depends on the accuracy of $P(X_1, \ldots ,X_n)$ and does not necessarily converge towards the optimal solution when applied on estimated probabilities.  

Let us now drop the assumption that the values of $\pi_i$'s are non-equal. That is,  $\pi_n \leq  \pi_{n-1} \leq  \cdots  \leq  \pi_1 \leq 1/2$. It may be easily verified that both
Theorem \ref{theorem1} and our suggested algorithm still hold, with the difference that instead of choosing the \textit{single} smallest entry in which the probability lists differ, we may choose \textit{one of the} (possibly many) smallest entries. This means that instead of recovering the unique value of $\pi_k$ at the $k^{th}$ iteration (as the values of $\pi_i$'s are assumed non-equal), we recover the $k^{th}$ smallest value in the list  $\pi_n \leq  \pi_{n-1} \leq  \cdots  \leq  \pi_1 \leq 1/2$.   

\subsection{Generalized BICA via Search Tree Based Algorithm}

We now turn to the general form of our problem (\ref{eq:min_criterion}) with no further assumption on the vector $\underline{X}$.

We denote $\Pi_i$= \{all words whose $i^{th}$ bit equals 0\}. In other words, $\Pi_i$ is the set of words that ``contribute" to $P(Y_i=0)$. We further denote the set of $\Pi_i$ that each word is a member of as $\Gamma_k$ for all $k=1\ldots N$ words. For example, the all zeros word $\{00\ldots0\}$ is a member of all $\Pi_i$ hence $\Gamma_1=\{ \Pi_1,\ldots,\Pi_n\}$. 

We define the optimal permutation as the permutation of the $N$ words that achieves the minimum of $C(\underline{Y})$ such that $\pi_i \leq 1/2$ for every $i$. 

Let us denote the binary representation of the $k^{th}$ word with $\underline{y}(k)$.
Looking at the $N$ words of the vector $\underline{Y}$ we say that a word $\underline{y}(k)$ is majorant to $\underline{y}(l)$ $(\underline{y}(k) \succeq \underline{y}(l))$ if  $\Gamma_l \subset \Gamma_k$. 
In other words, $\underline{y}(k)$ is majorant to $\underline{y}(l)$ iff for every bit in $\underline{y}(l)$ that equals zeros,  the same bit equals zero in $\underline{y}(k)$.  
In the same manner a word $\underline{y}(k)$ is minorant to $\underline{y}(l)$  $(\underline{y}(k)\preceq \underline{y}(l))$ if  $\Gamma_k \subset \Gamma_l$, that is iff for every bit in $\underline{y}(k)$ that equals zeros, the same bit equals zero in $\underline{y}(l)$. For example, the all zeros word $\{00\ldots0\}$ is majorant to all the words, while the all ones word $\{11\ldots1\}$ is minorant to all the word as non of its bits equals zeros.

We say that $\underline{y}(k)$ is a largest minorant to $\underline{y}(l)$ if there is no other word that is minorant to $\underline{y}(l)$ and majorant to $\underline{y}(k)$. We also say that there is a partial order between $\underline{y}(k)$  and $\underline{y}(l)$ if one is majorant or minorant to the other.  

\begin{theorem}  \label{th:partial_order}
The optimal solution must satisfy $P(\underline{y}(k))\geq P(\underline{y}(l))$ for all $\underline{y}(k)\preceq \underline{y}(l)$.
\end{theorem}

\begin{proof} 
Assume there exists $\underline{y}(k) \preceq \underline{y}(l), k \neq l$ such that $P(\underline{y}(k))<P(\underline{y}(l))$ which achieves the lowest (optimal) $C(\underline{Y})$. Since $\underline{y}(k) \preceq \underline{y}(l)$ then, by definition, $\Gamma_k \subset \Gamma_l$.  This means there exists $\Pi_j$ which satisfies $\Pi_j \in \Gamma_l \setminus \Gamma_k$. Let us now exchange (swap) the words $\underline{y}(k)$ and $\underline{y}(l)$. Notice that this swapping only modifies $\Pi_j$ but leaves all other $\Pi_i$'s untouched. Therefore this swap leads to a lower $C (\underline{Y})$ as the sum in (\ref{eq:sum_ent_min_binary}) remains untouched apart from its $j^{th}$ summand which is lower than before. This contradicts the optimality assumption $\blacksquare$
\end{proof}

We are now ready to present our algorithm. As a preceding step let us sort the probability vector $\underline{p}$ (of \underline{X}) such that $p_i \leq p_j$  for all $i<j$.
As described above, the all zeros word is majorant to all words and the all ones word is minorant to all words. Hence, the smallest probability $p_1$ and the largest probability $p_N$ are allocated to them respectively, as Theorem \ref{th:partial_order} suggests. We now look at all words that are largest minorants to the all zeros word.

Theorem \ref{th:partial_order} guarantees that $p_2$ must be allocated to one of them. We shall therefore examine all of them. This leads to a search tree structure in which every node corresponds to an examined allocation of $p_i$. In other words, for every allocation of $p_i$ we shall further examine the allocation of $p_{i+1}$ to each of the largest minorants that are still not allocated.  This process ends once all possible allocations are examined.

The following example (Figure \ref{fig:exhaustive}) demonstrates our suggested algorithm with $n=3$. The algorithm is initiated with the allocation of $p_1$ to the all zeros word. In order to illustrate the largest minorants to $\{000\}$ we use the chart of the partial order at the bottom left of Figure \ref{fig:exhaustive}. As visualized in the chart, a word $\underline{y}(l)$ is minorant to $\underline{y}(k)$ if it lies within one of the shapes (e.g. ellipses, rectangles) that also contains $\underline{y}(k)$. Therefore, the largest minorants to $\{000\}$ are $\{001\}$, $\{010\}$ and $\{100\}$. As we choose to investigate the allocation of $p_2$ to $\{001\}$ we notice that resulting largest minorants that are still not allocated are $\{010\}$ and $\{100\}$. We then investigate the allocation of $p_3$ to $\{010\}$, for example, and continue until all $p_i$ are allocated.  

\begin{figure}[h!]
\centering
\includegraphics[width = 0.6\textwidth,bb= 50 469 550 720,clip]{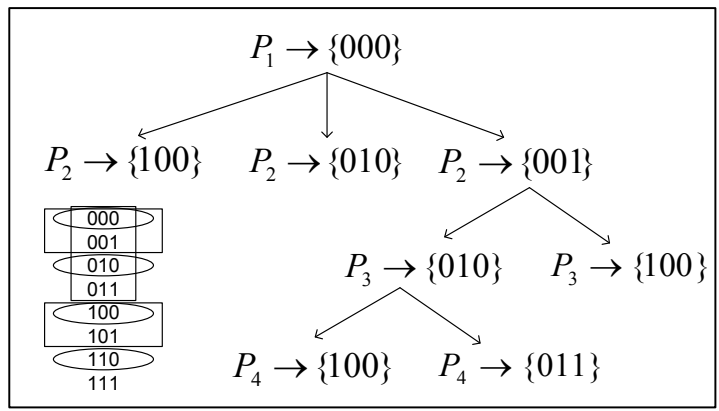}
\caption{Search tree based algorithm with $n=3$}
\label{fig:exhaustive}
\end{figure}

This search tree structure can be further improved by introducing a depth-first branch and bound enhancement. This means that before we examine a new branch in the tree we bound the minimal objective it can achieve (through allocation of the smallest unallocated probability to all of its unallocated words for example).  

The asymptotic computational complexity of this branch and bound search tree is quite difficult to approximate. There are, however, several cases where a simple solution exists (for example, for $n=2$ it is easy to show that the solution is to allocate all four probabilities in ascending order).

\subsection{Generalized BICA via Piecewise Linear Relaxation Algorithm}
\label{The Relaxed Generalized BICA}

In this section we present a different approach which bounds the optimal solution from above as tightly as we want in $O(n^k \cdot 2^n)$, where $k$ defines how tight the bound is. Throughout this section we assume that $k$ is a fixed value, for complexity analysis purposes.  

Let us first notice that the problem we are dealing with (\ref{eq:sum_ent_min_binary}) is a concave minimization problem over a discrete permutation set which is NP hard. However, let us assume for the moment that instead of our ``true" objective (\ref{eq:sum_ent_min_binary}) we have a simpler linear objective function. That is, 

\begin{equation} \label{eq:linear}
{\displaystyle L(\underline{Y})=\sum_{i=1}^{n}{a_i \pi_i+b_i}=\sum_{i=1}^{N}{c_i P(\underline{Y}=\underline{y}(i))}+d}
\end{equation}
where the last equality changes the summation over $n$ bits to a summation over all $N=2^n$ words (this change of summation is further discussed in Section \ref{matrix-vector}).

In order to minimize this objective over the $N$ given probabilities $\underline{p}$ we simply sort these probabilities in a descending order and allocate them such that the largest probability goes with the smallest coefficient $c_i$ and so on. Assuming both the coefficients and the probabilities are known and sorted in advance, the complexity of this procedure is linear in $N$.

We now turn to the generalized binary ICA problem as defined in (\ref{eq:sum_ent_min_binary}). Since our objective is concave we would first like to bound it from above with a piecewise linear function which contains $k$ pieces, as shown in Figure \ref{fig:piecewise linear}. In this paper we do not discuss the construction of such upper-bounding piecewise linear function, nor tuning methods for the value of $k$, and assume this function is given for any fixed $k$. Notice that the problem of approximating concave curves with piecewise linear functions is very well studied (for example \cite{gavrilovic1975optimal}) and may easily be modified to the upper bound case. 
We show that solving the piecewise linear problem approximates the solution to (\ref{eq:sum_ent_min_binary}) as closely as we want, in significantly lower complexity. 

\begin{figure}[h]
\centering
\includegraphics[width = 0.50\textwidth,bb= 100 260 520 520,clip]{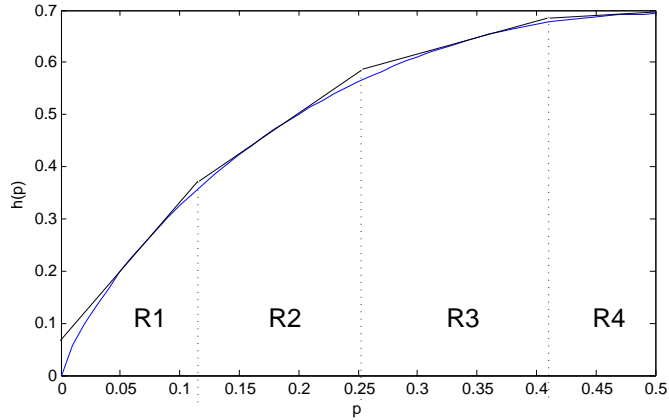}
\caption{piecewise linear ($k=4$) relaxation to the binary entropy}
\label{fig:piecewise linear}
\end{figure}

From this point on we shall drop the previous assumption that $\pi_n \leq  \pi_{n-1} \leq  \cdots  \leq  \pi_1$, for simplicity of presentation. 
First, we notice that all $\pi_i's$ are equivalent (in the sense that we can always interchange them and achieve the same result). This means we can find the optimal solution to the piecewise linear problem by going over all possible combinations of “placing” the $n$ variables $\pi_i$ in the $k$ different regions of the piecewise linear function. For each of these combinations we need to solve a linear problem (such as in (\ref{eq:linear}), where the minimization is with respect to allocation of the  $N$ given probabilities $\underline{p}$) with additional constraints on the ranges of each $\pi_i$. For example, assume $n=3$ and the optimal solution is such that two $\pi_i's$ (e.g. $\pi_1$ and $\pi_2$) are at the first region, $R_1$, and $\pi_3$ is at the second region, $R_2$. Then, we need to solve the following constrained linear problem,

\begin{equation} \label{eq:linear_program}
\begin{aligned}
& {\text{minimize}}
& &  a_1 \cdot (\pi_1+\pi_2)+2b_1+a_2\cdot \pi_3 +b_2 \\
& \text{subject to}
& & \pi_1, \pi_2 \in R_1 , \pi_3 \in R_2
\end{aligned}
\end{equation}

where the minimization is over the allocation of the given $\left\{p_i\right\}_{i=1}^{N}$, which determine the corresponding $\pi_i$'s, as demonstrated in (\ref{eq:linear}). This problem is again NP hard. However, if we attempt to solve it without the constraints we notice the following:

\begin{enumerate}

\item	If the collection of $\pi_i's$ which define the optimal solution to the unconstrained linear problem happens to meet the constraints then it is obviously the optimal solution with the constraints.
\item	If the collection of $\pi_i's$ of the optimal solution do not meet the constraints (say, $\pi_2 \in R_2$) then, due to the concavity of the entropy function, there exists a different combination with a different constrained linear problem (again, over the allocation of the $N$ given probabilities $\underline{p}$), 

\begin{equation*} \label{eq:linear_program}
\begin{aligned}
& {\text{minimize}}
& &  a_1\pi_1+b_1+a_2(\pi_2+\pi_3)+2b_2 \\
& \text{subject to}
& & \pi_1 \in R_1  \;  \pi_2,\pi_3 \in R_2
\end{aligned}
\end{equation*}

in which this set of $\pi_i's$ necessarily achieves a lower minimum (since $a_2 x+b_2<a_1 x+b_1$ $\forall x \in R_2$).

\end{enumerate}
	
Therefore, in order to find the optimal solution to the piecewise linear problem, all we need to do is to go over all possible combinations of placing the $\pi_i's$ in the $k$ different regions, and for each combination solve an unconstrained linear problem (which is solved in a linear time in $N$). If the solution does not meet the constraint then it means the assumption that the optimal $\pi_i$ reside within this combination's regions is false. Otherwise, if the solution does meet the constraint, it is considered as a candidate for the global optimal solution. 

The number of combinations is equivalent to the number of ways of placing $n$ identical balls in $k$ boxes, which is (for a fixed $k$), 

\begin{equation}
\left(\begin{array}{c} 
n+k-1\\n\end{array}\right)=\left(\begin{array}{c} 
n+k-1\\k-1\end{array}\right) \leq \frac{(n+k-1)^{k-1}}{(k-1)!} = O(n^k).
\end{equation}

Assuming the coefficients are all known and sorted in advance, for any fixed $k$ the overall asympthotic complexity of our suggested algorithm, as $n \rightarrow \infty$, is just $O(n^k \cdot 2^n)$.

\subsection{The Relaxed Generalized BICA as a single matrix-vector multiplication}
\label{matrix-vector}
It is important to notice that even though the asymptotic complexity of our approximation algorithm is $O(n^k\cdot2^n)$, it takes a few seconds to run an entire experiment on a standard personal computer for as much as $n=10$, for example. The reason is that the $2^n$ factor refers to the complexity of sorting a vector and multiplying two vectors,  
operations which are computationally efficient on most available software. Moreover, if we assume that the coefficients in (\ref{eq:linear}) are already calculated, sorted and stored in advance, we can place them in a matrix form $A$ and multiply the matrix with the (sorted) vector $\underline{p}$. The minimum of this product is exactly the solution to the linear approximation problem. Therefore, the practical complexity of the approximation algorithm drops to a single multiplication of a ($n^k \times 2^n$) matrix with a ($2^n \times 1$) vector

Let us extend the analysis of this matrix-vector multiplication approach. 

Each row of the matrix $A$ corresponds to a single coefficient vector to be sorted and multiplied with the sorted probability vector $\underline{p}$. 
Each of these coefficient vectors correspond to one possible way of placing $n$ bits in $k$ different regions of the piecewise linear function. Specifically, in each row, each of the $n$ bits is assumed to reside in one of the $k$ regions, hence it is assigned a slope $a_i$ as indicated in (\ref{eq:linear}). For each row, our goal is to minimize $L(\underline{Y})$. Since this minimization is solved over the vector $\underline{p}$ we would like to change the summation accordingly. To do so, each entry of the coefficient vector (denoted as $c_i$ in (\ref{eq:linear})) is calculated by summing all the slopes that correspond to each $\pi_i$. For example, let us assume $n=3$ where $\pi_1,\pi_2 \in R_1$, with a corresponding slope $a_1$ and intercept $b_1$, while the $\pi_3 \in R_2$ with $a_2$ and $b_2$. We use the following mapping: $P(\underline{Y}=000)=p_1, P(\underline{Y}=001)=p_2, \dots , P(\underline{Y}=111)=p_8$. Therefore

\begin{equation}
\begin{array}{c} 
\pi_1=P(Y_1=0)=p_1+p_2+p_3+p_4\\
\pi_2=P(Y_2=0)=p_1+p_2+p_5+p_6\\
\pi_3=P(Y_3=0)=p_1+p_3+p_5+p_7\\
\end{array}.
\end{equation}

The corresponding coefficients $c_i$ are then the sum of rows of the following matrix

\begin{equation}
A=\left(
\begin{array}{ccc} 
a_1&a_1& a_2\\
a_1& a_1& 0\\
a_1& 0& a_2\\
a_1& 0& 0\\
0& a_1& a_2\\
0& a_1& 0\\
0& 0& a_2\\
0& 0& 0\\
\end{array}\right).
\end{equation}

This leads to a minimization problem 
\begin{align}
 &  L(\underline{Y})=\sum_{i=1}^{n}{a_i \pi_i+b_i}=\nonumber \\ 
& a_1(\pi_1+\pi_2)+a_2\pi_3+2b_1+b_2= \\[0.5em] \nonumber
& (2a_1+a_2)p_1+2a_1p_2+(a_1+a_2)p_3+a_1p_4+(a_1+a_2)p_5+a_1p_6+a_2p_7+2b_1+b_2 \nonumber
\end{align}

where the coefficients of $p_i$ are simply the sum of the $i^{th}$ row in the matrix $A$.

Now let us assume that $n$ is greater than $k$ (which is usually the case). It is easy to see that many of the coefficients $c_i$ are actually identical in this case. Precisely, let us denote by $l_i$ the number of assignments for the $i^{th}$ region. Then, the number of unique $c_i$ coefficients is simply 
\begin{equation*}
\prod_{i=1}^{k}(l_i+1)-1
\end{equation*}
subject to $\sum_{i=1}^{k}{l_i}=n$.
Since we are interested in the worst case (of all rows of the matrix $A$), we need to find the 
 non-identical coefficients. This is obtained when $l_i$ is as "uniform" as possible. Therefore we can bound the number of non-identical coefficients from above by temporarily dropping the assumption that $l_i$'s are integers and letting $l_i=\frac{n}{k}$ so that

\begin{equation}
 \max{\prod_{i=1}^{k}(l_i+1)} \leq \left(\frac{n}{k}+1\right)^{k}=O(n^k).
 \end{equation}

This means that instead of sorting the $2^n$ coefficients for each row of the matrix $A$, we only need to sort $O(n^k)$ coefficients. 

Now, let us further assume that the data is generated from some known parametric model. In this case, some probabilities $p_i$ may also be identical, so that the probability vector $\underline{p}$ may also not require $O(n \cdot 2^n)$ operations to be sorted. For example, if we assume a block independent structure, such that $n$ components (bits) of the data are generated from $\frac{n}{r}$ independent and identically distributed blocks of of size $r$, then it can be shown that the probability vector $\underline{p}$ contains at most
\begin{equation}
 \left(\begin{array}{c} 
\frac{n}{r}+2^r-1\\\frac{n}{r}\end{array}\right)=O\left({\left(\frac{n}{r}\right)}^{2^r}\right)
\end{equation}
non-identical elements $p_i$. Another example is a first order stationary symmetric  Markov model. In this case there only exists a quadratic number, $ n\cdot(n-1)+2=O(n^2) $, of non-identical probabilities in $\underline{p}$ (see Appendix).

This means that applying our relaxed Generalized BICA on such datasets may only require $O(n^k)$ operations for the matrix $A$ and a polynomial number of operations (in $n$) for the vector $\underline{p}$; hence our algorithm is reduced to run in a polynomial time in $n$. 

Notice that this derivation only considers the number of non-identical elements to be sorted through a quicksort algorithm. However, we also require the degree of each element (the number of times it appears) to eventually multiply the matrix $A$ with the vector $\underline{p}$. This, however, may be analytically derived through the same combinatorical considerations described above.

\subsection{Relaxed Generalized BICA Illustration and Experiments}

In order to validate our approximation algorithm we conduct several experiments. 

In the first experiment we randomly generate a probability distribution with $n=10$ statistically independent components and mix its components in a non-linear fashion. We apply the approximation algorithm on this probability distribution with different values of $k$ and compare the approximated minimum entropy we achieve (that is, the result of the upper-bound piecewise linear cost function) with the entropy of the vector. In addition, we apply the estimated parameters $\pi_i$ on the true objective (\ref{eq:sum_ent_min_binary}), to obtain an even closer approximation. Figure \ref{fig:experiment} demonstrates the results we achieve, showing the convergence of the approximated entropy towards the real entropy as the number of linear pieces increases.  Moreover, as we repeat this experiment several times (that is, randomly generate probability distributions and examine our approach for every single value of $k$), we see that the estimated parameters are equal to the independent parameters for $k$ as small as $4$, on average.
     
\begin{figure}[h]
\centering
\includegraphics[width = 0.50\textwidth,bb= 130 265 470 515,clip]{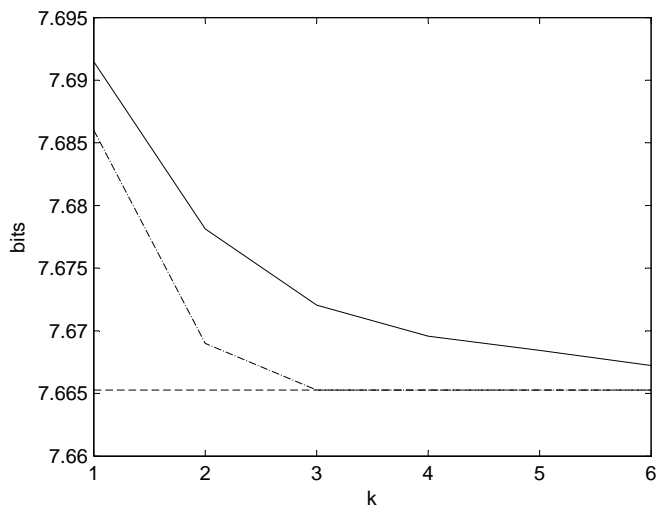}
\caption{Piecewise linear approximation (solid line), entropy according to the estimated parameters (dashed-dot line) and the real entropy (dashed horizontal line),  for a vector size $n=10$ and different $k$ linear pieces}
\label{fig:experiment}
\end{figure}

We further illustrate the use of the BICA tool by the following example on ASCII code. The ASCII code is a common standardized eight bit representation of western letters, numbers and symbols. We gather statistics on the frequency of each character, based on approximately 183 million words that appeared in the New York Times magazine \cite{jones2004case}. We then apply the BICA (with $k=8$, which is empirically sufficient) on this estimated probability distribution, to find a new eight bit representation of characters, such that the bits are "as statistically independent" as possible. We find that the entropy of the joint probability distribution is 4.8289 bits, the sum of marginal entropies using ASCII representation is 5.5284 bits and the sum of marginal entropies after applying BICA is just 4.8532 bits. This means that there exists a different eight bit representation of characters which allows nearly full statistical independence of bits. Moreover, through this representation one can encode each of the eight bit separately without losing more than 0.025 bits, compared to encoding the eight bits altogether.

\section{Generalized Independent Component Analysis Over Finite Alphabets} 
\label{GICA over FF}
\subsection{Piecewise Linear Relaxation Algorithm - Exhaustive Search}
\label{Exhaustive Search}
Let us extend the notation of the previous sections, denoting the number of components as $n$ and the alphabet size as $q$. We would like to minimize $\sum_{i=1}^{n}{H(Y_i)}$ where $Y_i$ is over an alphabet size $q$.
We first notice that we need $q-1$ parameters to describe the multinomial distribution of $Y_i$ such that all of the parameters are not greater than $\frac{1}{2}$. Therefore, we can bound from above the marginal entropy with a piecewise linear function in the range $[0,\frac{1}{2}]$, for each of the parameters of $Y_i$.  
We call a ($q-1$)-tuple of regions a \textit{cell}.
As in previous sections we consider $k$, the number of linear pieces, to be fixed. Notice however, that as $q$ and $n$ increase, $k$ needs also to take greater values in order to maintain the same level of accuracy. As mentioned above, in this work we do not discuss methods to determine the value of $k$ for given $q$ and $n$, and empirically evaluate it. 

Let us denote the number of cells to be visited in our approximation algorithm (Section \ref{The Relaxed Generalized BICA}) as $C$. Since each parameter is approximated by $k$ linear pieces and there are $q-1$ parameters, $C$ equals at most $k^{q-1}$.

In this case too, the parameters are exchangeable (in the sense that the entropy of a multinomial random variable with parameters $\{p_1,p_2,p_3\}$ is equal to the entropy of a multinomial random variable with parameters $\{p_2,p_1,p_3\}$, for example). Therefore, we do not need to visit all $k^{q-1}$ cells, but only a unique subset which disregards permutation of parameters. In other words, the number of cells to be visited is bounded from above by the number of ways of choosing $q-1$ elements (the parameters) out of $k$ elements (the number of pieces in each parameter) with repetition and without order. Notice this upper-bound (as opposed to full equality) is a result of not every combination being a feasible solution, as the sum of parameters may exceed $1$. Assuming $k$ is fixed and as $q \rightarrow \infty$ this equals 

\begin{equation}
\label{q^k}
\left(\begin{array}{c} 
q-1+k-1\\q-1\end{array}\right)=\left(\begin{array}{c} q-1+k-1\\k-1\end{array}\right)\leq\frac{(q-1+k-1)^{k-1}}{(k-1)!}=O(q^k).
\end{equation}
Therefore, the number of cells we are to visit is simply $C=\min\left(k^{q-1},O(q^k)\right)$. For sufficiently large $q$ it follows that $C=O\left(q^k\right)$.
As in the binary case we would like to examine all combinations of $n$ entropy values in $C$ cells. The number of iterations to calculate all possibilities is equal to the number of ways of placing $n$ identical balls in $C$ boxes, which is

\begin{equation}
\label{n^C}
\left(\begin{array}{c} 
n+C-1\\n\end{array}\right)=O\left(n^C\right).
\end{equation}

In addition, in each iteration we need to solve a linear problem which takes a linear complexity in $q^n$.
Therefore, the overall complexity of our suggested algorithm is $O\left(n^{C}\cdot q^n \right)$.

We notice however that for a simple case where only two components are mixed $(n=2)$, we can calculate (\ref{n^C}) explicitly
\begin{equation}
\left(\begin{array}{c} 
2+C-1\\2\end{array}\right)=\frac{C(C+1)}{2}.
\end{equation}

Putting this together with (\ref{q^k}), leads to an overall complexity which is polynomial in $q$, for a fixed $k$,

\begin{equation}
\left(\frac{q^k(q^k+1)}{2}q^2\right)=O\left(q^{2k+2}\right).
\end{equation}

Either way, the computational complexity of our suggested algorithm may result in an excessive runtime, to a point of in-feasibility, in the case of too many components or an alphabet size which is too large. This necessitates a heuristic improvement to reduce the runtime of our approach.

\subsection{Piecewise Linear Relaxation Algorithm - Objective Descent Search}
\label{objective descent search}
In Section \ref{The Relaxed Generalized BICA} we present the basic step of our suggested piecewise linear relaxation to the generalized binary ICA problem. As stated there, for each combination of placing $n$ components in $k$ pieces (of the piecewise linear approximation function) we solve a linear problem (LP). Then, if the solution happens to meet the constraints (falls within the ranges we assume) we keep it. Otherwise, due to the concavity of the entropy function, there exists a different combination with a different constrained linear problem in which this solution that we found necessarily achieves a lower minimum, so we disregard it.

This leads to the following objective descent search method: instead of searching over all possible combinations we shall first guess an initial combination as a starting point (say, all components reside in a single cell). We then solve its unconstrained LP. If the solution meets the constraint we terminate. Otherwise we visit the cell that meets the constraints of the solution we found. We then solve the unconstrained LP of that cell and so on. We repeat this process for multiple random initialization. 

\begin{algorithm}
\caption{Relaxed Generalized ICA For Finite Alphabets via Gradient Search}
\begin{algorithmic} [1]
\REQUIRE $\underline{p}$ = the probability function of the random vector $\underline{X}$
\REQUIRE $n$ = the number of components of $\underline{X}$
\REQUIRE $C$ = the number of cells which upper-bound the objective.
\REQUIRE $I$ = the number of initializations.
\STATE $opt \leftarrow \infty$, where the variable $opt$ is the minimum sum of marginal entropies we are looking for. 
\STATE $V \leftarrow \emptyset$, where $V$ is the current cells the algorithm is visiting.
\STATE $S\leftarrow \infty$, where $S$ is the solution of the current LP.

\STATE $i \leftarrow 1$.
\WHILE{$i \leq I$}

\STATE $V\leftarrow$ randomly select an initial combination of placing $n$ components in $C$ cells 
\STATE $S\leftarrow$ LP($V$) solve an unconstrained linear prograsm which corresponds to the selected combination, as appears in (\ref{eq:linear}). \label{marker}
\IF {the solution falls within the bounds of the cell} 
\IF {$H(S)<opt$} 
\STATE $opt\leftarrow H(S)$, the sum of marginal entropies which correspond to the parameters found by the LP 
\ENDIF
\STATE $i\leftarrow i+1$
\ELSE 
\STATE $V \leftarrow$ the cells in which $S$ reside.   
\STATE \textbf{goto 7}
\ENDIF
\ENDWHILE
\RETURN $opt$
\end{algorithmic}
\label{alg:algorithm}
\end{algorithm}

This suggested algorithm is obviously heuristic, which does not guarantee to provide the global optimal solution. Its performance strongly depends on the number of random initializations and the concavity of the searched domain. 

The following empirical evaluation demonstrates our suggested approach.
In this experiment we randomly generate a probability distribution with $n$ independent and identically distributed components over an alphabet size $q$. We then mix its components in a non-linear fashion. We apply the objective descent algorithm with a fixed number of initialization points ($I=1000$) and compare the approximated minimum sum of the marginal entropies with the true entropy of the vector.
Figure \ref{fig:experiment_q4} demonstrates the results we achieve for different values of $n$. We see that the objective descent algorithm approximates the correct components well for smaller values of $n$ but as $n$ increases the difference between the approximated minimum and the optimal minimum increases, as the problem becomes too overwhelming. 

\begin{figure}[h]
\centering
\includegraphics[width = 0.45\textwidth,bb= 50 180 550 590,clip]{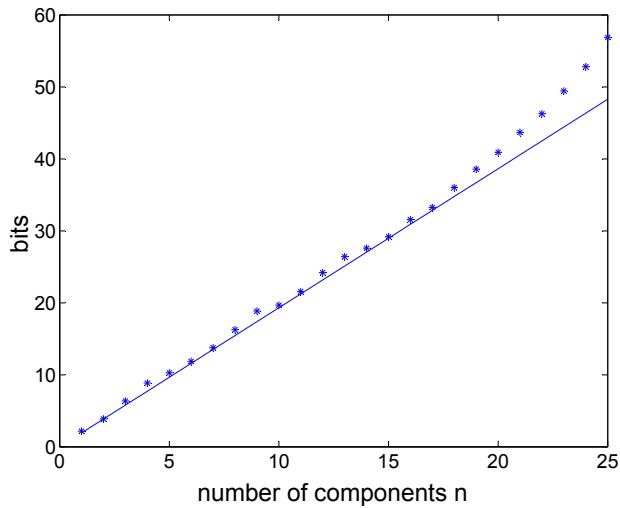}
\caption{The real entropy (solid line) and the sum of marginal entropies as discovered by the objective descent algorithm, for an i.i.d vector over an alphabet size $q=4$ and of varying number of components $n$}
\label{fig:experiment_q4}
\end{figure}

\section{Constrained Generalized Independent Component Analysis Over Finite Alphabets} 
\label{Constrained BICA}
The outcome of the Generalized ICA, regardless of its computational cost, is a mapping from all $q^n$ words of $\underline{X}$ to $q^n$ words in $\underline{Y}$, $g:\{1,\dots,q\}^n \rightarrow \{1,\dots,q\}^n$, where $q$ is the alphabet size of each component and $n$ is the number of components. This mapping alone may be too large to store in a database, not to mention retrieve and use. Therefore, we may be interested in a sub-optimal mapping that is more compact. One way of defining such a mapping is through imposing a constraint on the maximal number of components $X_i$ that every $Y_i$ depends on. This enables a simpler, function-like mapping that is easier and cheaper to implement on both hardware or software.

In this section we turn our attention back to the binary case, $q=2$. Recall that a balanced boolean function \cite{chakrabarty1998balanced} is a binary function whose output yields as many $0s$ as $1s$ over its input set. This means that for a uniformly random input string of bits, the probability of getting a $1$ is $\frac{1}{2}$.

We notice that each component of $\underline{Y}$ must hold as many $0s$ as $1s$, as the transformation is invertible (every possible word in $\underline{X}$ is mapped to a different word in $\underline{Y}$). This means that a necessary condition for $\underline{Y}$ to be invertible is that each component $Y_i$ is a balanced boolean function (of as many as $r$ components) of $\underline{X}$.

\subsection{Functions of Two Components ($r=2$)}
\label{r=2}
Let us first focus on a simple case where $r=2$. It is easy to verify that in this case, every balanced function is linear over $GF(2)$. Notice that this no longer applies for $r>2$, as the set of linear functions is just a subset of all balanced functions.
Therefore, for $r=2$, the solution is much easier to attain. For example, if we assume a structure where $Y_i=f_i(x_i,x_{i+1\pmod{n}})$ then there exist at most $2^{n+1}-2$ different invertible transformations to search, in order to find the optimal solution (see Appendix). Since the size of this search space is exponential in $n$, we may want to apply a heuristic based on Attux et al. \cite{attux2011immune}, to find a sub-optimal solution to this problem in fewer operations. 

In their work, Attux et al. suggest an immune-inspired algorithm for minimizing (\ref{eq:sum_ent_min}) over XOR field (linear transformations).  Their algorithm starts with a random "population" where each element in the population represents a valid transformation ($W$, an invertible matrix). At each step, the affinity function evaluates the objective $\left(1-\frac{1}{n}\sum_{i=1}^{n}{H(Y_i)}\right)$ for each element in the population, which is subsequently cloned. Then, the clones suffer a mutation process that is inversely proportional to their affinity, generating a new set of individuals. This new set is evaluated again in order to select the individual with highest affinity, for each group of clone individuals and their parent individual. The process is finished with a random generation of $d$ new individuals to replace the lowest affinity individuals in the population. The entire process is repeated until a pre-configured  number of repetitions is executed. Then, the solution with the highest affinity is returned. It is important to notice that the mutation phase is implemented with a random bit resetting routine, with the constraint of accepting only new individuals that form a nonsingular matrix (invertible transformation).

Therefore, we can simply adjust Attux et al.'s algorithm to satisfy our constraints by forcing each row in $W$ to hold at most $r$ elements equal to $1$. This change shall be embedded in the mutation phase, right before we verify that $W$ is non-singular. 

We notice that we can also apply this heuristic for cases where $r>2$. However, the solution we get is limited to linear transformations, which is a substantially smaller domain than all invertible transformations. 

\subsection{Functions of Multiple Components ($r>2$)}

In the case where $Y_i$ is a function of $r>2$ components we can no longer restrict ourselves to linear transformations in order to find the optimal transformation. 
We notice two major restrictions that define our search space:

\begin{enumerate}

\item	The transformation needs to be invertible.
\item	$Y_i$ is a function of no more than $r$ components of $\underline{X}$.

\end{enumerate}

We denote all the transformations in this search space as ``feasible" solutions. 
The two constraints mentioned dramatically reduce the number of feasible solutions to our optimization problem, yet they define a highly non-trivial search domain; given a feasible solution it is not clear how to find a different feasible solution. This makes the problem much more difficult to tackle and necessitates an efficient methodology for visiting all feasible solutions. 

There are two immediate ways to search through all feasible solutions. The first is going over all transformations that hold the first constraint (invertible transformations) and for each examining if it holds the second constraint. The second option is to go over all transformations for which the second constraint holds and keep only those that are also invertible.
Since there exist $2^n$ invertible transformations, the first option is infeasible. The following Branch-and-Feasible algorithm is based on the second option:

Let us first choose one of ${n \choose \sfrac{n}{2}}$ balanced boolean functions for $Y_1$, such that $Y_1$ is a function of no more than $r$ components of $\underline{X}$. Then, we shall again choose a balanced boolean function (of no more than $r$ components) for $Y_2$. However, not all functions are permitted as some of them may result in a non invertible transformation already at this step. This way, we grow branches of feasible balanced boolean functions. We notice that as we get deeper into the tree there exist fewer permitted splits (as it is harder to find functions that result in invertible transformations). Moreover, some choices of functions may result in branches that already exist in the tree. These branches shall obviously be pruned. 
A detailed discussion regarding the properties of this algorithm and its practical abilities are subject to future research.  

\section{Applications}
\subsection{Blind Source Separation}

We start this section with a classical ICA application, in Blind Source Separation (BSS). Assume there exist $n$ independent (or practically "almost independent") sources where each source is over an alphabet size $q$. These sources are mixed in an invertible, yet unknown manner. Our goal is to recover the sources from the mixture. 

For example, consider a case with $n=2$ sources $X_1,X_2$, where each source is over an alphabet size $q$. The sources are linearly mixed (over a finite field) such that $Y_1=X_1, Y_2=X_1+X_2$. 

However, due to a malfunction, the symbols of $Y_2$ are randomly shuffled, before it is handed to the receiver.
Notice this mixture (including the malfunction) is unknown to the receiver, who receives $Y_1,Y_2$ and strives to ``blindly" recover $X_1,X_2$. 
 In this case any linearly based method such as \cite{yeredor2011independent} or \cite{attux2011immune} would fail to recover the sources as the mixture, along with the malfunction, is now a non-linear invertible transformation. Our method on the other hand, is designed especially for such cases, where no assumption is made on the mixture (other than being invertible).

To demonstrate this example we introduce two independent sources $X_1,X_2$, over an alphabet size $q$. We apply the linear mixture $Y_1=X_1, Y_2=X_1+X_2$ and shuffle the symbols of $Y_2$. We are then ready to apply (and compare) our suggested methods for finite alphabet sizes, which are the exhaustive search method (Section \ref{Exhaustive Search}) and the objective descent method (Section \ref{objective descent search}).  

For the purpose of this experiment we assume both $X_1$ and $X_2$ are distributed according to a Zipf's law distribution, 
\begin{equation}\nonumber
P(k;s,q)=\frac{k^{-s}}{\sum_{m=1}^q m^{-s}}
\end{equation}
with a parameter $s=1.6$. The Zipf's law distribution is a commonly used heavy-tailed distribution. This choice of distribution is further motivated in Section \ref{large alphabet}.

We apply our suggested algorithms for different alphabet sizes, with a fixed $k=8$, and with only $100$ random initializations for the objective descent method. Figure \ref{fig:BSS} presents the results we achieve.

\begin{figure}[h]
\centering
\includegraphics[width = 0.7\textwidth,bb= 50 125 800 480,clip]{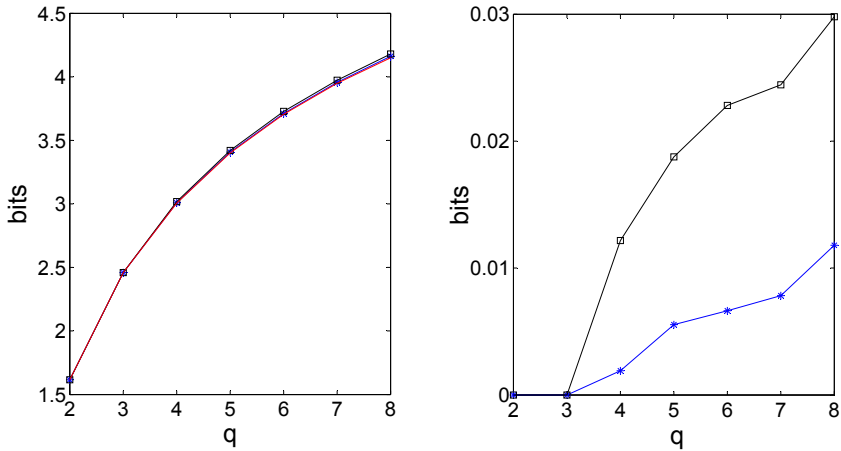}
\caption{BSS simulation results. Left: the lower curve is the joint entropy of $Y1,Y2$, the asterisks curve is the sum of marginal entropies using the exhaustive search method (Section \ref{Exhaustive Search}) while the curve with the squares corresponds the objective descent method (Section \ref{objective descent search}). Right: the curve with the asterisks corresponds to the difference between the exhaustive search method and the joint entropy while the curve with the squares is the difference between the objective descent search method and the joint entropy.}
\label{fig:BSS}
\end{figure}

We first notice that both methods are capable of finding a transformation for which the sum of marginal entropies is very close to the joint entropy. This means our suggested methods succeed in separating the non-linear mixture $Y_1,Y_2$ back to the statistically independent sources $X_1,X_2$, as we expected.  

Looking at the chart on the right hand side of Figure \ref{fig:BSS}, we notice that the difference between the two methods tends to increase as the alphabet size $q$ grows. This is not surprising since the search space grows while the number of random initializations remains fixed.  However, the difference between the two methods is still practically negligible, as we can see from the chart on the left. This is especially important since the objective descent method takes significantly less time to apply as the alphabet size $q$ grows.

\subsection{Source Coding - Generalization of Predictive Coding}
In his pioneering work from 1955, Peter Elias \cite{elias1955predictive} presented a sequential scheme for source coding. He showed that instead of coding a process directly, the encoder should first predict the current value of the process from its past realizations and then code the prediction error. The prediction errors are encoded component-wise, as if they are statistically independent. This way the encoder strives to eliminate any statistical dependency and only codes the information that is still not known. 

We generalize this procedure by saying that instead of sequential generation of statistically independent components (prediction errors), an improved encoder shall assemble a vector of observations and make it "as statistically independent as possible" all together. This can be shown to achieve improved performance (closer to the entropy of the vector) at the expense of working in batch. Though our suggested finite group ICA algorithms are of high complexity, it is important to notice that they are to be applied once, off-line, in order to generate the optimal mapping $g:\{0,1\}^n \rightarrow \{0,1\}^n$.

Compared with Huffman coding, the main advantage of the predictive coding scheme is that unlike sequential Huffman coding, it does not require an exponentially increasing codebook and utilizes a fixed codebook for each component it codes. It does however require an exponentially increasing “prediction book” so that each realization of past observations is mapped to a single prediction of the present. This problem may be evaded by limiting the prediction module to a parametric structure, as discussed in Section \ref{Constrained BICA}.

\subsection{Distributed Source Coding}
Let ${(x_1 (i),x_2 (i))}_{i=1}^{\infty}$ be a sequence of independent and identically distributed drawings of a pair of correlated discrete random variables $X_1$ and $X_2$. For lossless compression with $\hat{X}_1=X_1$ and $\hat{X}_2=X_2$ after decompression, we know from Shannon's source coding theory that a rate given by the joint entropy $H(X_1,X_2)$ of $X_1,X_2$ is sufficient if we are encoding them together, as in Figure \ref{fig:DSC1}.a. For example, we can first compress $X_1$ into $H(X_1)$  bits and based on the complete knowledge of $X_1$ at both the encoder and the decoder, we then compress $X_2$ into $H(X_2|X_1)$ bits. But what if $X_1$ and $X_2$ must be separately encoded for some user to reconstruct them together? 

One simple way is to do separate coding with rate $R=H(X_1)+H(X_2)$, which is greater than $H(X_1,X_2)$ when $X_1$ and $X_2$ are correlated. In a landmark paper \cite{slepian1973noiseless}, Slepian and Wolf showed that $R=H(X_1,X_2)$ is sufficient even for separate encoding of correlated sources (Figure \ref{fig:DSC1}.b). The Slepian--Wolf theorem says that the achievable region of Distributed Source Coding (DSC) for discrete sources $X_1$ and $X_2$ is given by $R_1\geq H(X_1|X_2)$, $R_2 \geq H(X_2|X_1)$ and $R_1+R_2\geq H(X_1,X_2)$, which is shown in Figure \ref{fig:DSC2}. The proof of the Slepian--Wolf theorem is based on random binning. Binning is a key concept in DSC and refers to partitioning the space of all possible outcomes of a random source into disjoint subsets or bins. However, the Slepian--Wolf approach suffers from several drawbacks. First, just as in Shannon's channel coding theorem, the random binning argument used in the proof of the Slepian--Wolf theorem is asymptotic and non-constructive; In general, Slepian--Wolf codes are not easy to implement in practice. Second, as a result of non-random binning applied in most known practical codes, a complete zero-error regime is not applicable and some error is evident. Third, the Slepian--Wolf theorem is not easily scalable to a greater number of sources and non-binary alphabets.

Let us now consider a case in which we are interested in a highly scalable distributed source coding system, in which the encoders are “cheap”, in the sense that they consist of a single predefined codebook and cannot be reconfigured post deployment. Such a setup applies, for example, where an extremely large number of sources need to be encoded in a distributed manner by multiple encoders, in order to reduce the computational cost required by each encoder. In addition, we would like the system to be lossless in the strict sense, such that the decompressed random variables are equal to the original sources for any block length. 
We notice that if the sources are statistically independent of each other then the rate of each encoder is simply given by the marginal entropy of that source, and the overall rate achieves Shannon's joint entropy lower bound.  Inspired by this, we ask ourselves whether there exists a pre-processing component which makes the sources as “statistically independent as possible”, prior to the distributed compression. As Figure \ref{fig:DSC1}.c demonstrates, we would like to find an invertible transformation $\underline{Y}=g(\underline{x})$ such that $\underline{Y}$ is of the same dimension (and the same alphabet size). In addition we would like the components of $\underline{Y}$ to be as statistically independent as possible in the sense that $\sum{H(Y_i)} \rightarrow \min$, where $H(Y_i)$ corresponds to the minimal rate of the $i^{th}$ encoder and $\sum{H(Y_i)}$ is bounded by $H(\underline{X})$, the joint entropy of all sources. This is exactly achieved by our suggested generalized ICA over finite group algorithms, where both the transmitter and the receiver are assumed to know $P_{\underline{X}}$, and are therefore capable of finding the corresponding mapping (and its inverse).         

\begin{figure}[h]
\centering
\includegraphics[width = 0.5\textwidth,bb= 145 395 450 665,clip]{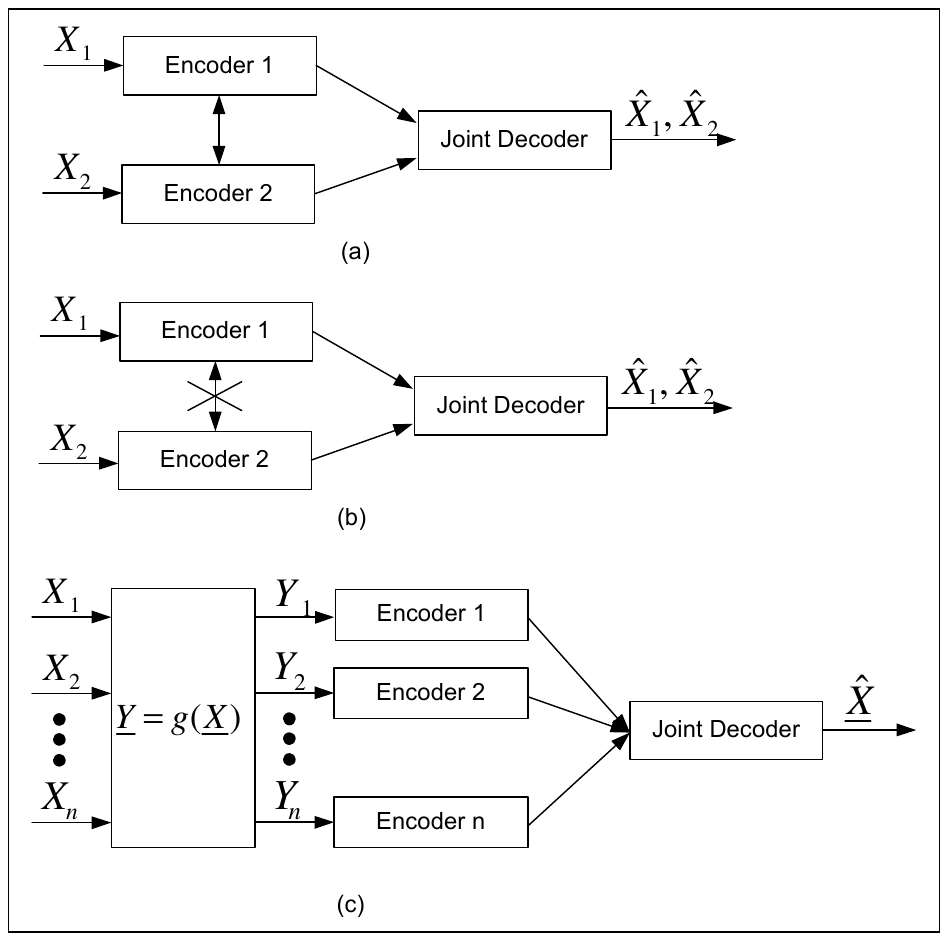}
\caption{(a) Joint encoding of $X_1, X_2$. The encoders collaborate and a rate $H(X_1,X_2)$ is sufficient. (b) Distributed encoding of $X_1, X_2$. The encoders do not collaborate. The Slepian--Wolf theorem says that a rate $H(X_1,X_2)$ is also sufficient provided that the decoding of $X_1$ and $X_2$ is done jointly. (c) Distributed encoding with a pre-processing stage which enables a scalable, zero-error system with cheap fixed codebook encoders working in a rate “close” to $H(X_1,X_2)$.}
\label{fig:DSC1}
\end{figure}

However, the pre-processing stage is a somewhat centralized component. That is, one may claim that by sensing all sources $\underline{X}$, the centralized component can communicate to each encoder the realizations of the other sources. This shall enable conditional coding as in the centralized source coding scheme described above. This kind of regime, however, entails “expensive” and sophisticated encoders which are able to communicate and store multiple codebooks (exponential in the number of sources).    

\begin{figure}[h]
\centering
\includegraphics[width = 0.5\textwidth,bb= 185 495 400 662,clip]{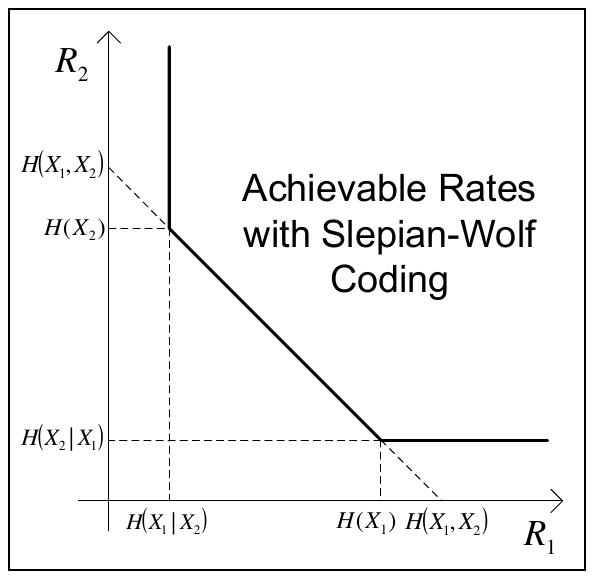}
\caption{The Slepian--Wolf rate region for two sources}
\label{fig:DSC2}
\end{figure}

Still, assume we are interested in eliminating, or at least reducing, centralized components. This means we would like each component $Y_i$ to be dependent on as few components of $\underline{X}$ as possible. Obviously we cannot simply apply our generalized ICA algorithms on $\underline{X}$ as it may result with sophisticated dependencies between $Y_i$ and the components of $\underline{X}$. However, we may restrict these dependencies by applying the constrained version of our method, discussed in Section \ref{Constrained BICA}. For example, we may define a spatial pattern in which each component $Y_i$ depends on no more than $r$ adjacent components of $\underline{X}$, such that $Y_i=f_i(x_i,\dots,x_{i+r\pmod{n}})$. This pattern is analyzed for $r=2$ in Section \ref{r=2} and can be extended to $r>2$ as later discussed in that section. 

This method of restricting the number of components each encoder is responsible for, followed by generalized ICA mapping, can be viewed as a trade-off between how centralized/decoupled the encoder is and the achievable rate which can be achieved in a lossless source coding regime for any block length.

\subsection{Source Coding of Large Alphabet Data}
\label{large alphabet}
Let ${\underline{x}(i)}_{i=1}^{N}$ be a sequence of $N$ i.i.d. vector realizations we are to encode. Assume the dimension of the vector $\underline{x}$ is $n$ and each of its components take values over an alphabet size $q$. Therefore, the alphabet size of the vectors $\underline{x}(i)$ is $q^n$. 

A compressed representation of a dataset involves two components - the compressed data itself and an overhead redundancy (for example, in the form of a dictionary for decompression purpose). Encoding the data requires at least $N$ times the data's empirical entropy, $\hat{H}(\underline{X})$. This is attained through entropy coding according to the data's empirical distribution.
The redundancy, on the other hand, may be quantified in several ways. One common way of measuring the coding redundancy is through the minimax criterion \cite{davisson1973universal}. Szpankowski and Weinberger \cite{szpankowski2012minimax} derived the minimax point-wise redundancy for various ranges of sequence length and alphabet size. Throughout this section we use their results to quantify the encoding redundancy.

Large alphabet source coding considers the case where the alphabet size $q^n$ is of order larger than the number of realizations $N \ll q^n$. Unlike the case where $N \gg q^n$, here the coding redundancy may not be negligible. Precisely, the minimax point-wise redundancy behaves asymptotically as

 \begin{equation}
\label{large_ab}
d_{N,q^n}\simeq N\log{\frac{q^n}{N}}.
\end{equation}

We argue that in some setups it is better to split the components of the data into blocks, with $b$ components in each block, and encode the blocks separately. Notice that we may set the value of $b$ so that the blocks are no longer considered as over a large alphabet size $(N \gg q^b)$. This way, the redundancy of encoding each block separately is again negligible, at the cost of longer averaged codeword length. For simplicity of notation we define the number of blocks as $B$, and assume $B=\frac{n}{b}$ is a natural number. Therefore, according to (\ref{large_ab}) encoding the $n$ components all together takes a total of about  

\begin{equation}
\label{compression_large}
N\cdot \hat{H}(\underline{X})+N\log{\frac{q^n}{N}}
\end{equation}
bits, while the block-wise compression takes about

\begin{equation}
\label{compression_blocks}
N\cdot \sum_{j=1}^{B}{ \hat{H}({\underline{X}}^{(j)})}+B\frac{q^b-1}{2}\log{\frac{N}{q^b}}
\end{equation}
bits, where the first term is $N$ times the sum of $B$ block entropies and the second term is $B$ times the redundancy of each block when $(N\gg q^b)$.

Two subsequent questions arise from this setup:
\begin{enumerate}

\item	What is the optimal value of $b$ that minimizes (\ref{compression_blocks})?
\item	Given a fixed value of $b$, how can we re-arrange $n$ components into $B$ blocks so that the averaged codeword length (lower-bounded by the empirical entropy), together with the redundancy, is as small as possible?

\end{enumerate}

Let us start by fixing $b$ and focusing on the second question. 

A naive approach is to exhaustively or randomly search for all possible combinations of clustering $n$ components into $B$ blocks. Assuming $n$ is quite large, an exhaustive search is practically infeasible; hence a different method is required.
We suggest applying our generalized ICA tool as an upper-bound search method for efficiently searching for a minimal possible averaged codeword length.
  
As in previous sections we define $\underline{Y}=g(\underline{X})$, where $g$ is some invertible transformation of $\underline{X}$.
Every block of the vector $\underline{Y}$ satisfies 
\begin{equation}
\label{upperbound}
\sum_{i=1}^{b}{\hat{H}({Y}_{i}^{(j)})} \geq \hat{H}(\underline{Y}^{(j)})
\end{equation}
where $\underline{Y}^{(j)}$ refers to the $j^{th}$ block of size $b$. This means that the sum of marginal empirical entropies of each block bounds the empirical entropy of the block from above. Summing over all $B$ blocks in (\ref{upperbound}) we have 

\begin{equation}
\label{upperbound2}
 \sum_{j=1}^{B}{\sum_{i=1}^{b}{\hat{H}({Y}_{i}^{(j)})}}=
 \sum_{i=1}^{n}{\hat{H}(Y_i)}  \geq  \sum_{j=1}^{B}{\hat{H}({\underline{Y}}^{(j)})}
\end{equation}
which indicates that the sum of the block empirical entropies is upper bounded by the marginal empirical entropies of the components of $\underline{Y}$ (with equality iff the components are independently distributed).

  

Our suggested scheme works as follows: 
We first randomly partition the $n$ components into $B$ blocks. We estimate the joint probability of each block and apply the generalized ICA on it. The sum of marginal empirical entropies (of each block) is an upper bound on the empirical entropy of each block, as indicated in (\ref{upperbound}). Now, let us randomly permute the $n$ components of the vector $\underline{Y}$. Here by "permute" we refer to an exchange of positions of $\underline{Y}$'s components, as opposed to permuting the words, like we did in previous sections. By doing so, the sum of marginal empirical entropies of the entire vector $\sum_{i=1}^{n}{\hat{H}(Y_i)}$ is maintained. However, as we now apply the generalized ICA on each of the (new) blocks, we necessarily minimize (or at least do not increase) the sum of marginal empirical entropies of the (new) block, hence minimize the sum of marginal empirical entropies of the entire vector $\underline{Y}$. This means that we minimize the upper bound of the sum of block empirical entropies, as (\ref{upperbound2}) suggests. In other words, we show that in each iteration we decrease (at least do not increase) an upper bound on our objective. We terminate once a maximal number of iterations is reached or we can no longer decrease the sum of marginal empirical entropies.

Algorithm 2 summarizes our suggested scheme.

\begin{algorithm}
\caption{Large Alphabet Source Coding via Generalized ICA over Finite Alphabets}
\begin{algorithmic} [1]
\REQUIRE $\underline{X}$=a vector of realizations $N$.
\REQUIRE $n$ = the dimension of $\underline{X}$. 
\REQUIRE $B$= the number of blocks.
\REQUIRE $I$ = the number of iterations.

\STATE $\hat{H}_{m}(1)\leftarrow \sum_{i=1}^{n}{\hat{H}(Y_i)}$ , where $\hat{H}_{m}(i)$ is the sum of marginal empirical entropies of the vector $\underline{X}$ at iteration $i$. 

\STATE $eps\leftarrow 10^{-10}$
\STATE $i\leftarrow 2$, where $i$ be an iterations counter.
\WHILE {$i \leq I$ and $\hat{H}_{m}(i) \leq \hat{H}_{m}(i-1)$}
\STATE Randomly partition $n$ components into $B$ blocks

\FOR {$j=1$ to $B$}

\STATE $\hat{H}^{(j)}(i) \leftarrow$  the empirical entropy of the $j^{th}$ block at iteration $i$.
\STATE Apply Generalized ICA over Finite Alphabets on the $j^{th}$ block 
\STATE $\hat{H}_{m}^{(j)}(i)\leftarrow $ the sum of marginal empirical entropies of the $j^{th}$ block at iteration $i$. 

\ENDFOR
\STATE $\hat{H}_{m}(i) \leftarrow \sum_{j=1}^{B}{\hat{H}_{m}^{(j)}(i)}$.
\STATE $\hat{H}_{b}(i) \leftarrow \sum_{j=1}^{B}{\hat{H}^{(j)}(i)}$, where  $\hat{H}_{b}(i)$ is the sum of block empirical entropies of the vector $\underline{X}$ at iteration $i$.
\STATE $i\leftarrow i+1$
\ENDWHILE
\RETURN $\left\{\hat{H}_{b}(i)\right\}_{i=1}^I$
\end{algorithmic}
\label{alg:algorithm}
\end{algorithm}

Therefore, encoding the data if we terminate at iteration $I_0$ takes a total of about

\begin{align}
\label{compression_total}
   N\cdot \sum_{j=1}^{B}{\hat{H}^{[I_0]}({\underline{Y}}^{(j)})}+B\frac{q^b-1}{2}\log{\frac{N}{q^b}}+ I_0B\cdot &bq^b+I_0n\log{n}
\end{align}
where the first term refers to the sum of block empirical entropies at the $I_0$ iteration, the third term refers to the representation of $I_0 \cdot B$ invertible transformation of each block during the process until $I_0$, and the fourth term refers to the bit permutations at the beginning of each iteration.

Hence, in order to minimize (\ref{compression_total}) we need to find the optimal trade-off between a low value of $\sum_{j=1}^{B}{\hat{H}^{[I_0]}({\underline{Y}}^{(j)})}$ and a low iteration number $I_0$. 
We may then apply this technique with different values of $b$ to find the best compression scheme over all block sizes.
 
In order to demonstrate our suggested method we generate a dataset according to Zipf distribution. The Zipf distribution is a discrete distribution commonly used for modeling of natural (real-world) quantities. It is widely used in physical and social sciences, linguistics, economics and many other fields. 
We draw $N=10^6$ realizations from the Zipf distribution with a parameter $1.4$ and represent each realization through a $24$ bit representation. For this dataset we attain an empirical entropy of $5.55$ bits. Therefore, compressing the drawn realizations in its given $2^{24}$ alphabet size takes a total of $10^{6}\times5.55+10^6\times\log{\frac{2^{24}}{10^6}}=9.62 \cdot 10^6$ bits, according to (\ref{compression_large}). 

Let us now apply a block-wise compression. 
We first demonstrate the behavior of our suggested approach with four blocks $(B=4)$ in Figure \ref{fig:zipf}, compared with a naive search which randomly searches for all possible shufflings of $n$ components into $B$ blocks. (as described above).

\begin{figure}[h]
\centering
\includegraphics[width = 0.45\textwidth,bb= 50 190 550 590,clip]{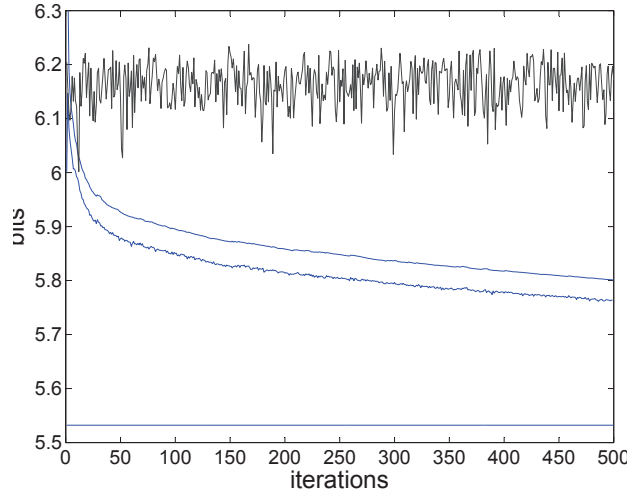}
\caption{Large Alphabet Source Coding via Generalized ICA over Finite Alphabets. The horizontal line indicated the empirical entropy of $\underline{X}$. The upper curve is the sum of marginal empirical entropies which upper bounds the lower curve, the sum of block empirical entropies (the outcome of our suggested approach at each iteration). The fluctuating line is the sum of block empirical entropies as returned by a naive random search}
\label{fig:zipf}
\end{figure}

We first notice that the naive search is able to find a minimum of $\sum_{j=1}^{B}{\hat{H}({\underline{Y}}^{(j)})}=6$ bits. This leads to a total of about $6 \cdot 10^6$ bits for the entire dataset, which is significantly lower than the $9.62 \cdot 10^6$ bits we achieve through a single block compression. It is important to mention that this minimum is found while randomly searching over a much larger number of iterations than indicated in Figure (\ref{fig:zipf}).

Applying our suggested scheme, we minimize (\ref{compression_total}) over $I_0=50$ and $\sum_{j=1}^{B}{\hat{H}({\underline{Y}}^{(j)})}=5.85$  to achieve a total of $5.93 \cdot 10^6$ bits for the entire dataset, saving about  $7 \cdot 10^4$ bits compared to the naive search. Table \ref{table:zipf_results} summarizes the results we achieve for different block sizes $B$, compared with a naive search.  

\begin{table*}[t]
\caption{Comparison Between Block-Wise Compression via a Naive Search and a Generalized ICA Method. "Our approach" refers to a Large Alphabet Source Coding via Generalized ICA over Finite Alphabets while the naive search refers to a random search result, when clustering $n$ components into $B$ blocks (that is, the minimum over $1000$ random trails)}
\renewcommand{\arraystretch}{1.7} 
\label{table:zipf_results} 
\centering
\begin{tabular}{c|c|c|c|c|c|c}  
\hline

\begin{tabular}{@{}c@{}}Number of \\ Blocks\end{tabular}
&\begin{tabular}{@{}c@{}}Number of Components\\ in each Block \end{tabular}
&\begin{tabular}{@{}c@{}} Naive Search Minimum of \\ $\sum_{j=1}^{B}{\hat{H}({\underline{Y}}^{(j)})}$\end{tabular}
&\begin{tabular}{@{}c@{}} Naive Search \\ Total Compression\end{tabular}
&\begin{tabular}{@{}c@{}} Our Approach \\ Optimal $I_0$\end{tabular}
&\begin{tabular}{@{}c@{}} Our Approach  \\ $\sum_{j=1}^{B}{\hat{H}({\underline{Y}}^{(j)})}$\end{tabular}
&\begin{tabular}{@{}c@{}} Our Approach \\ Total Compression\end{tabular}
\\[1ex]  \hline
\hline
$6$   & $4$ & $6.1$& $6.1\cdot 10^6$ & $70$ & $6$ & $\bold{6.03\cdot10^6}$  \\[1ex]  \hline
$4$  & $6$ & $6$& $6\cdot 10^6$ & $50$ & $5.85$ & $\bold{5.93\cdot10^6}$  \\[1ex]  \hline
$3$  & $8$ & $5.9$& $5.9\cdot 10^6$ & $10$ &$5.8$ & $\bold{5.86\cdot10^6}$\\[1ex]  \hline

\end{tabular} 
\end{table*}

\section{Conclusion}
In this work we consider a generalized ICA over finite alphabets framework where we drop the common assumptions on the underlying model. Specifically, we attempt to decompose a given "categorical" vector to its "as statistically independent as possible" components with no further assumptions, as introduced by Barlow \cite{barlow1989finding}.  

We first focus on the binary case and propose three algorithms to address this class of problems. In the first algorithm we assume that there exists a set of independent components that were mixed to generate the observed vector. We show that these independent components are recovered in a combinatorial manner in $O(n \cdot 2^n)$ operations. The second algorithm drops this assumption and accurately solves the generalized BICA problem through a branch and bound search tree structure. Then, we propose a third algorithm which bounds our objective from above as tightly as we want to find an approximated solution in $O(n^k \cdot 2^n)$ with $k$ being the approximation accuracy parameter. We further show that this algorithm can be formulated as a single matrix-vector multiplications and under some generative model assumption the complexity is dropped to be polynomial in $n$. 

Following that we extend our methods to deal with a larger alphabet size. This case necessitates a heuristic approach to deal with the super exponentially increasing complexity. An objective descent search method is presented for that purpose. 

In addition we show that in some applications we may get a reduced complexity using a constrained version of the generalized ICA framework. This version incorporates additional restrictions on the number of components $r$ in $\underline{X}$ that each element of $\underline{Y}$ is dependent on. This guarantees that the outcome of the generalized ICA algorithm is not an exponentially increasing mapping (in $n$) but a simpler and more compact structure. We show this problem can be solved through known heuristics when $r=2$ but requires a branch and feasible approach for larger values of $r$.     

The code for our suggested methods in publically available \footnote{\url{https://sites.google.com/site/amichaipainsky}} for future research. 

We conclude the paper by presenting four applications, three of which are in the field of source coding. The first source coding application extends the predictive coding scheme. We show that our generalized ICA method may apply to this setup to enhance its performance. This means that instead of sequentially generating statistically independent components (prediction errors) we generate a (more) statistically independent vector at the cost of introducing a time delay. Second, we show that our method applies for distributed source coding. We claim that by making the source "as statistically independent as possible" prior to encoding each component separately, we achieve lossless compression for any block size while minimizing the minimal rate we may work at. We conclude by discussing the problem of encoding data over a large alphabet size. Our suggested scheme, based on the generalized ICA method, shows to efficiently find a transformation that achieves a low entropy for a block based entropy encoder. 

Although we focus our interest on source coding, the generalized ICA over finite alphabets method is a universal tool which aims to decompose data into its fundamental components. We believe that the theoretical properties we introduce, together with the practical algorithms which utilize a variety of setups, make this tool applicable to many disciplines.

\appendix

\begin{theorem}
Assume a binary random vector $\underline{X} \in \{0,1\}^n$ is generated from a first order stationary symmetric  Markov model. Then, the joint probability of $\underline{X}$,  $P_{\b{x}}=p_1, \dots, p_n$ only contains $n\cdot(n-1)+2$ unique (non-identical) elements of $p_1, \dots, p_n$.
\end{theorem}

\begin{proof}
We first notice that for a binary, symmetric and stationary Markov model, the probability of each word is solely determined by 

\begin{enumerate}

\item	The value of the first (most significant) bit 
\item	The number of elements equal $1$ (or equivalently $0$) 
\item	The number of transitions from $1$ to $0$ (and vice versa).

\end{enumerate}
 
For example, for $n=4$ the probability of $0100$ equals the probability of $0010$, while it is not equal to the probability of $0001$.

Assume the number of transitions, denoted in $r$, is even. Further, assume that the first (most significant) bit equals zero. Then, the number of words 
with a unique probability is 
\begin{equation}
\label{U1}
U_1=\sum_{r=2,\; r\; \text{is even}}^{n-2} \sum_{k=\frac{r}{2}}^{n-\frac{r}{2}}1=\sum_{r=2,\; r\; \text{is even}}^{n-2}n-r
\end{equation}

where the summation over $r$ corresponds to the number of transitions, while the summation over $k$ is with respect to the number of $1$ elements given $r$. For example, for $n=4$ and $r=2$ we have the words $0100, 0010$ for $k=1$ (which have the same probability as discussed above), and the word $0110$ for $k=2$. 
In the same manner, assuming again that the most significant bit is $0$ but now $r$ is odd, we have 
\begin{equation}
\label{U2}
U_2=\sum_{r=2,\; r\; \text{is odd}}^{n-1} \sum_{k=\frac{r+1}{2}}^{n-\frac{r+1}{2}}1=\sum_{r=2,\; r\; \text{is odd}}^{n-1}n-r.
\end{equation}

Putting together (\ref{U1}) and (\ref{U2}) we have that number of words with a unique probability, assuming the most significant bit is $0$, equals 

\begin{equation}
U1+U_2=\sum_{r=1}^{n-1} n-r = \frac{n\cdot(n-1)}{2}+1.
\end{equation}

The same derivation holds for the case where the most significant bit is $1$, leading to a total of $n\cdot(n-1)+2$ words with a unique probability $\blacksquare$

\end{proof}

\begin{theorem}
There are $2^{n+1}-2$ invertible $n \times n$ matrices over the binary field which satisfy
\begin{equation}\nonumber
\left\{A_{i,j}=0 \;\; \forall i,j\;\;  s.t.\;\; j\neq i \; ,\; j\neq{i+1}\pmod{n}\right\}.
\end{equation}

\end{theorem}
\begin{proof}
Let us first denote the diagonal of a matrix $A$,  $A_{i,i}$  as $d_1$ while $d_2$ refers to the secondary diagonal above it, $A_{i,i+1}$.
We denote the case where a diagonal $d_k$ is all ones as $d_k=1$.
 The only entries of $A$ which can take non-zero values are the entries on $d_1, d_2$ and the bottom left entry, $A_{n,1}$. 
A matrix is invertible iff its determinant is non-zero. Laplace's formula expresses the determinant of a matrix in terms of its minors. The minor $M_{i,j}$ is defined to be the determinant of the $(n-1) \times (n-1)$ matrix that results from $A$ by removing the $i^{th}$ row and the $j^{th}$ column. The determinant of $A$ is then given by

\begin{equation}
det(A)=\sum_{i=1}^{n}{(-1)^{i+j}A_{i,j}M_{i,j}}
\end{equation}
for any column $1\leq j\leq n$. As we limit ourselves to the binary field the sum operator is simply a XOR between two binary variables.
Let us derive the conditions for $\det(A) \neq 0$.
First, assume $A_{n,1}=0$. In this case $A$ is a triangular matrix and  $\det(A)=\prod_{i=1}^n A_{i,i}$. Therefore, the determinant of $A$ is not equal to zero iff the diagonal $d_1$ is all ones. Hence, we have $2^{n-1}$ values for $d_2$ such that $\det(A) \neq 0$ in this case.

In the case where $A_{n,1}=1$ we have $\det(A)=A_{1,1}M_{1,1}+M_{n,1}$. This means $\det(A) \neq 0$ iff either $A_{1,1}M_{1,1}=1$ or $M_{n,1}=1$, but not both of them together. 
We already found that $A_{1,1}M_{1,1}=1$ iff $d_1$ is all ones and in the same way $M_{n,1}=1$ iff $d_2$ is all ones. Therefore $\det(A) \neq 0$ iff either $d_1=1$ or $d_2=1$.
This means we have $\left(2^{n-1}-1 \right)+\left(2^{n}-1 \right)$ different matrices which are invertible in this case.

Putting this all together we have $2^{n-1}+\left(2^{n}-1 \right)+\left(2^{n-1}-1 \right)=2^{n+1}-2$ different matrices which satisfy the conditions of $A$ and are invertible $\blacksquare$
\end{proof}


%



\section*{Acknowledgment}
This research was supported in part by a returning scientists grant
to Amichai Painsky from the Israeli Ministry of Science, and by Israeli Science
Foundation grants 1487/12 and 634/09. The authors thank Greg Ver Steeg for helpful comments regarding previous work in this field. The thorough comments of an anonymous reviewer are particularly appreciated.

\bibliographystyle{IEEEtran}
\bibliography{bibi}


\begin{IEEEbiographynophoto}{Amichai Painsky}
received his B.Sc. degree in Electrical Engineering from Tel Aviv University (2007) and his M.Eng. degree in Electrical Engineering from Princeton University (2009). He is currently carrying a Ph.D. at the Statistics department of Tel Aviv University School of Mathematical Sciences. His research interests include Data Mining, Machine Learning, Statistical Learning and their connection to Information Theory
\end{IEEEbiographynophoto}

\begin{IEEEbiographynophoto}{Saharon Rosset}
is an Associate Professor in the department of Statistics and Operations Research at Tel Aviv University. His research interests are in Computational Biology and Statistical Genetics, Data Mining and Statistical Learning. Prior to his tenure at Tel Aviv, he received his PhD from Stanford University in 2003 and spent four years as a Research Staff Member at IBM Research in New York. He is a five-time winner of major data mining competitions, including KDD Cup (four times) and INFORMS Data Mining Challenge, and two time winner of the best paper award at KDD (ACM SIGKDD International Conference on Knowledge Discovery and Data Mining)
\end{IEEEbiographynophoto}

\begin{IEEEbiographynophoto}{Meir Feder}
(S'81-M'87-SM'93-F'99) received the B.Sc and M.Sc degrees
from Tel-Aviv University, Israel and the Sc.D degree from the Massachusetts
Institute of Technology (MIT) Cambridge, and the Woods Hole
Oceanographic Institution, Woods Hole, MA, all in electrical
engineering in 1980, 1984 and 1987, respectively.

After being a research associate and lecturer in MIT he joined
the Department of Electrical Engineering - Systems, School of Electrical Engineering, Tel-Aviv
University, where he is now a Professor and the incumbent of the Information Theory Chair.
He had visiting appointments at the Woods Hole Oceanographic Institution, Scripps
Institute, Bell laboratories and has been a visiting
professor at MIT. He is also extensively involved in the high-tech
industry as an entrepreneur and angel investor.
He co-founded several companies including Peach Networks,
a developer of a server-based interactive TV solution which
was acquired by Microsoft, and Amimon a
provider of ASIC's for wireless high-definition A/V connectivity.

Prof. Feder is a co-recipient of the 1993 IEEE Information Theory
Best Paper Award. He also received the 1978 "creative thinking"
award of the Israeli Defense Forces, the 1994 Tel-Aviv University
prize for Excellent Young Scientists, the 1995 Research Prize of
the Israeli Electronic Industry, and the research prize in applied
electronics of the Ex-Serviceman Association, London, awarded by
Ben-Gurion University.
\end{IEEEbiographynophoto}
\vfill



\end{document}